\documentclass[envcountsame, runningheads]{llncs}

\usepackage[T1]{fontenc}
\usepackage{graphicx}

\usepackage[all]{xy}
\usepackage[mathscr]{eucal}
\usepackage{amssymb}
\usepackage{amsmath}
\usepackage{enumerate}
\usepackage{url} 
\usepackage{multicol}
\usepackage{marginnote}
\usepackage{comment}

\usepackage{pgf,tikz}

\usetikzlibrary{calc,positioning,shapes,arrows.meta}

\tikzset{%
 shaded/.style={draw, shape=circle, fill=black!35, inner sep=1.4pt},
 unshaded/.style={draw, shape=circle, fill=white, inner sep=1.4pt},
 quasi/.style={draw, shape=rectangle, rounded corners=3pt, fill=white, inner sep=2.5pt, minimum height=14.5pt},
 blob/.style={draw, shape=rectangle, rounded corners=12pt, thin, densely dotted},
 arrow/.style={->, thin, >=latex, shorten >=2.5pt, shorten <=2.5pt},
 order/.style={thin},
 curvy/.style={thin, looseness=1.2, bend angle=70},
 fatcurvy/.style={thin, looseness=1.7, bend angle=75},
 label/.style={shape=rectangle, inner sep=6pt},
 auto}

\begin{document}

\title{Pregroup representable expansions of residuated lattices}
\titlerunning{Pregroup representable expansions of residuated lattices}
%
\author{Andrew Craig \inst{1,2}\orcidID{0000-0002-4787-3760} 
 \and
Claudette Robinson\inst{1}\orcidID{0000-0001-7789-4880}}
\authorrunning{A. Craig, C. Robinson}
%
\institute{Department of Mathematics and Applied Mathematics, University of Johannesburg, Auckland Park 2006, South Africa \\
\email{\{acraig,claudetter\}@uj.ac.za}
\and
National Institute for Theoretical and Computational Sciences (NITheCS), Johannesburg, South Africa
}

\maketitle

\begin{abstract}
Group representable relation algebras play an important role in the study of representable relation algebras.
The class of distributive involutive FL-algebras (DInFL-algebras) generalises relation algebras, as well as  Sugihara monoids and MV-algebras.  
We construct DInFL-algebras  from pregroups and show that they can be represented as algebras of binary relations. 
Even for finite pregroups we obtain relational representations of DInFL-algebras with non-Boolean lattice reducts. 
If the pregroup is enriched with a particular unary order-reversing operation, then our construction yields representation results for distributive quasi relation algebras. 
\\

\end{abstract}

\section{Introduction}

A well-known result 
attributed to McKinsey (cf.~\cite{JT48}) says that the complex algebra of a group is a relation algebra. Moreover, such a relation algebra is representable as an algebra of binary relations~\cite{JT52}. Here we seek to give relational representations for distributive involutive FL-algebras (DInFL-algebras) and distributive quasi relation algebras (DqRAs).  DInFL-algebras are a class of residuated lattice expansions that include relation algebras, Sugihara monoids, and MV-algebras. We give precise definitions of DInFL-algebras and DqRAs in 
Section~\ref{sec:preliminaries}.

As a generalisation of the well-known group representable relation algebras, we show that DInFL-algebras constructed as the up-set algebra of a pregroup are \emph{representable} as algebras of binary relations in the sense of~\cite{RDqRA25} (see  Section~\ref{sec:representable}). 
Group representable relation algebras have been studied extensively and the results obtained there have contributed significantly to the understanding of representable relation algebras. Our generalisation provides a platform for further investigation of representations of DInFL-algebras. 

Pregroups were introduced by Lambek~\cite{Lam99} to provide an algebraic model of the grammar of natural language. They are both a  generalisation of groups (see Example~\ref{ex:group_is_pregroup}) and special partially ordered monoids. Thus they are well-suited for our purposes, as the collection of up-sets of a partially ordered monoids can be given the structure of a distributive residuated lattice (Proposition~\ref{prop:RLUPP}). 

All finite pregroups are  discretely ordered, and hence their up-set algebras have a Boolean lattice reduct. However, these up-set algebras have interesting subalgebras and hence we are able to obtain many examples of pregroup representable  DInFL-algebras (see Example~\ref{ex:represenatble_via_Z7} and Section~\ref{sec:apps}).

We remark that DInFL-algebras are term equivalent to distributive involutive semirings (where distributivity must be expressed without the meet operation)~\cite{JV21}. Hence any results obtained here can also be applied in that setting. 

In the next section we recall definitions of the algebras under consideration: DInFL-algebras and DqRAs. We then examine the partially ordered monoids (pomonoids) that will be used to construct the algebras: ipo-monoids and pregroups. In Section~\ref{sec:embedDRL} we consider distributive residuated lattices (DRLs) which are  built as the up-sets of pomonoids.  Proposition~\ref{prop:CondW}  characterises when these concrete DRLs can be embedded into algebras of binary relations. The embedding $\sigma$  is an adaptation of the map used for group representable relation algebras.

Section~\ref{sec:representable} recalls the definition of representability for DInFL-algebras and DqRAs from our recent paper~\cite{RDqRA25}. In Section~\ref{sec:Rep_DInFL} we  construct DInFL-algebras from ipo-monoids and then show how this leads to representations as algebras of binary relations. The pregroup properties are essential for the embedding to preserve the linear negations. 

Ortho pregroups are introduced in  Section~\ref{sec:DqRA-ortho} in order to prove representability results for DqRAs. 
In Section~\ref{sec:apps} we give an example of a collection of DqRAs for which our method gives representations, and lastly in Section~\ref{sec:future} we point to some potential directions for further research.

\section{Preliminaries}\label{sec:preliminaries}

\subsection{Residuated lattice expansions}\label{sec:RLEs}

We begin by recalling the definitions of the algebras that we study in this paper. They can be viewed either as generalisations of relation algebras or as expansions of residuated lattices~\cite{GJ13}. 

A {\em residuated lattice} (RL) is an algebra 
$\mathbf{A}=\langle A,\wedge,\vee, \cdot,1, \backslash,/\rangle$ such that 
$\langle A,\cdot,1\rangle$ is a monoid, $\langle A,\wedge,\vee\rangle$
is
a lattice and the monoid operation $\cdot$ has 
residuals $\backslash$ and $/$, i.e., for all $a,b,c\in A$,
\[a\cdot b\leqslant c\quad \textnormal{ iff } \quad 
a\leqslant c/b\quad \textnormal{ iff } \quad
b\leqslant a\backslash c.\]
For all monoid operations in this paper, we write $ab$ for $a\cdot b$ whenever no confusion can arise. An element $a$ in $A$ is \emph{idempotent} if $a^2 = a$, and a residuated lattice $\mathbf{A}$ is \emph{idempotent} if all its elements are  idempotent. 
We say that a residuated lattice $\mathbf A$  is 
\emph{distributive} if its underlying lattice is distributive. 


When a residuated lattice is expanded with a constant $0$, the algebra  
$\mathbf{A}=\langle A,\wedge,\vee,\cdot,1, \backslash,/,0\rangle$
is called a {\em Full Lambek }(\emph{FL}-)\emph{algebra} (cf.~\cite[Chapter~2.2]{GJKO}).  
Note that no additional properties are assumed about the constant $0$. 
On an FL-algebra, two  {\em linear negations},
are defined in terms of the residuals and $0$ as follows:  ${\sim} a=a\backslash 0$ and ${-}a=0/a$.
It follows that ${\sim}1 = 1\backslash 0 = 0= 0/1=-1$.
By residuation, 
${\sim} (a\vee b)={\sim}a \wedge {\sim}b$ and ${-}(a\vee b)=
{-a}\wedge{-}b$ for all $a,b\in A$.
An FL-algebra $\mathbf{A}$ is \emph{cyclic} if $-a = {\sim} a$ for all $a \in A$. Clearly an FL-algebra $\mathbf{A}$ that is  commutative will be cyclic. 
The reader is referred to~\cite{GJKO} 
for more information.

If an FL-algbra $\mathbf{A}$ satisfies the condition 
\[\textsf{(In)}:\quad   {\sim}{-}a=a={-}{\sim}a,\text{ for all }a\in A,\]
then it is called an {\em involutive Full Lambek {\normalfont(}InFL-{\normalfont)}algebra}.
These algebras satisfy 
$a\leqslant b \textnormal{ iff } a\,\cdot({\sim} b)\leqslant -1 \textnormal{ iff } (-b)\cdot a\leqslant-1$,
for all $a,b\in A$. Further, since both $-$ and $\sim$ are order-reversing, 
we get that 
$-$ and $\sim$ are dual lattice isomorphisms in an InFL-algebra. 


Galatos and Jipsen~\cite[Lemma 2.2]{GJ13} showed that 
an InFL-algebra is term-equivalent to an algebra $\mathbf{A}=\langle A, \wedge,\vee,\cdot,1, -, {\sim}\rangle$ such that $\langle A, \wedge,\vee\rangle$ is a lattice, $\langle A, \cdot, 1\rangle$ is a monoid, and for all $a, b, c \in A$, we have
\begin{equation}
a\cdot b\leqslant c\quad \textnormal{ iff } \quad 
a\leqslant {-}\left(b\cdot {\sim}c\right)\quad \textnormal{ iff } \quad 
b\leqslant {\sim}\left(-c\cdot a\right).\label{eq:InFL}
\end{equation}

where $0 = -1 ={\sim}1$. 
If needed, the residuals can be expressed in terms of $\cdot$ and the linear negations as follows: $ c/b = -\left(b\cdot {\sim}c\right)$ and 
$    a\backslash c = {\sim}\left(-c\cdot a\right)$.

An InFL-algebra $\mathbf{A}$ can be expanded with an additional
unary operation ${\neg}:A\to A$ 
to form an {\em InFL$'$-algebra} 
$\mathbf{A}=\langle A, \wedge,\vee,\cdot,1, -, \sim,{\neg}\rangle$
such that $\neg\neg a=a$ for all $a\in A$. 
The involution is sometimes denoted by $'$, but 
we will use $\neg$.

If an InFL$'$-algebra $\mathbf{A}$
 satisfies the De Morgan law \textsf{(Dm)} $\neg (a\vee b)=\neg a\wedge \neg b$, for all $a,b,\in A$, 
then $\mathbf{A}$ is called a {\em DmInFL$'$-algebra}.  
A {\em quasi relation algebra} (qRA) is a DmInFL$'$-algebra 
$\mathbf{A}=\langle A,\wedge,\vee,\cdot,1,-,{\sim},{\neg}\rangle$
that satisfies
\[
\textsf{(Dp)}:\quad\neg (a\cdot b)=
{\sim}(-\neg b \cdot -\neg a).
\]
The abbreviation 
\textsf{(Dp)} stands for
\emph{De Morgan product}.

If $\mathbf{A}$ is a qRA, then it can be shown that $\neg 1 = - 1= {\sim}1 = 0$. 
An equivalent definition of a qRA is an FL$'$-algebra (i.e. an FL-algebra with an involutive unary operation $\neg$) that satisfies \textsf{(Dm)}, \textsf{(Dp)} and  \textsf{(Di)}: $\neg({\sim} a)=-(\neg a)$. 

Finally, a {\em distributive quasi relation
algebra} (DqRA) is a quasi relation algebra  $\mathbf{A}=\langle A,\wedge,\vee,\cdot, 1, -, {\sim}, {\neg}\rangle$
where the underlying lattice $\langle A,\wedge,\vee\rangle$
is distributive. 

The following example shows explicitly how 
 (distributive) InFL-algebras and  (distributive) qRAs are generalisations of relations algebras. 
\begin{example}\label{ex:relation_algebras}
Let $\langle A,\wedge, \vee, \top, \bot,', \cdot, 1,  ^\smallsmile\rangle$ be a relation algebra. Define $-a={\sim}a = (a')^\smallsmile$ and $\neg a = a'$ for all $a \in A$. Then $\langle A,\wedge, \vee, \cdot, 1, -, {\sim}\rangle$ is a cyclic DInFL-algebra and $\langle A,\wedge, \vee, \cdot, 1, -, {\sim}, \neg\rangle$ is a cyclic DqRA. 
\end{example}

For a full list of finite 
DInFL-algebras and DqRAs up to size 8, we  refer the reader to~\cite{CJR-DqRAs}. We highlight the class of algebras below as the methods in this paper will in Section~\ref{sec:apps} show them to be representable (see Section~\ref{sec:representable}). 

\begin{example}\label{ex:represenatble_via_Z7}
Let $n \in \omega$ with $n \geqslant 3$. If $n = 2k$ for $k\geqslant 2$, set $$A_n = \{a_{-k}, a_{-k+1}, \ldots,a_{-1}, a_{1}, a_{2}, \ldots, a_{k}, b_{-k}, b_{-k+1}, \ldots,b_{-1}, b_{1}, b_{2}, \ldots, b_{k}\},$$
and if $n = 2k +1$ for $k \geqslant 1$, set
$$A_n = \{a_{-k},  \ldots,a_{-1}, a_0, a_{1}, a_{2}, \ldots, a_{k}, b_{-k}, \ldots,b_{-1}, b_0,  b_{1}, b_{2}, \ldots, b_{k}\}.$$
For $a_i, a_j \in A_n$, define $a_i\vee a_j = a_{\textnormal{max}\{i, j\}}$ and $a_i\wedge a_j = a_{\textnormal{min}\{i, j\}}$. For $b_i, b_j \in A_n$,  define $b_i\vee b_j = b_{\textnormal{max}\{i, j\}}$ and $b_i\wedge b_j = b_{\textnormal{min}\{i, j\}}$. For $a_i, b_j \in A_n$, define 
\vspace{-0.75cm}
\begin{multicols}{2}
\begin{enumerate}[]   
\item $$a_i \vee b_j= b_j \vee a_i =  \begin{cases}
a_j & \text{if }\, j > i \\ 
a_i & \text{if }\, j \leqslant i 
\end{cases}$$ 
\item
$$a_i \wedge b_j = b_j \wedge a_i =  \begin{cases}
b_i & \text{if }\, j > i \\ 
b_j & \text{if }\, j \leqslant i. 
\end{cases}$$
\end{enumerate}
\end{multicols}
\noindent
It is straightforward to check that $\langle A, \wedge, \vee\rangle$ is a lattice.

For $a_i, a_j, b_i, b_j \in A_n$, define 
\vspace{-0.75cm}
\begin{multicols}{2}
\begin{enumerate}[]  
\item $$a_i \cdot a_j =  \begin{cases}
a_i & \text{if }\, j = -k \\ 
a_j & \text{if }\, i =-k\\
a_k  & \text{if }\, i \neq -k \text{ and } j \neq -k.
\end{cases}$$
\item 
$$b_i \cdot b_j = 
\begin{cases}
b_{-k} & \text{if }\, i = -k \textnormal{ or } j = -k\\
b_k & \text{if }\, -k < i, j \leqslant 0 \\
a_k & \text{if }\, 0 < i \leqslant k \textnormal{ or } 0 < j \leqslant k.
\end{cases}$$
\end{enumerate}
\end{multicols}
\noindent
and
$$a_i\cdot  b_j = b_j \cdot a_i =
\begin{cases}
b_j & \text{if }\, i= -k
\textnormal{ or } j = -k\\
b_k & \text{if }\, -k < i \leqslant 0 \textnormal{ or } -k < j \leqslant 0 \\ 
a_k & \text{if }\,  0 < i \leqslant k \textnormal{ and } 0 < j \leqslant k.
\end{cases}$$
If we set $1 = a_{-k}$, then $\langle A_n, \cdot, 1\rangle$ is a monoid. 

Finally, for $a_i \in A_n$, define $-a_i = {\sim}a_i = \neg a_i = b_{-i}$, and for $b_j \in A_n$, define $-b_j = {\sim}b_j = \neg b_j = a_{-j}$. It can then be shown that $\mathbf{A}_n=\langle A_n,\wedge, \vee, \cdot, 1, -, {\sim}, \neg\rangle$ is a (cyclic) DqRA. 

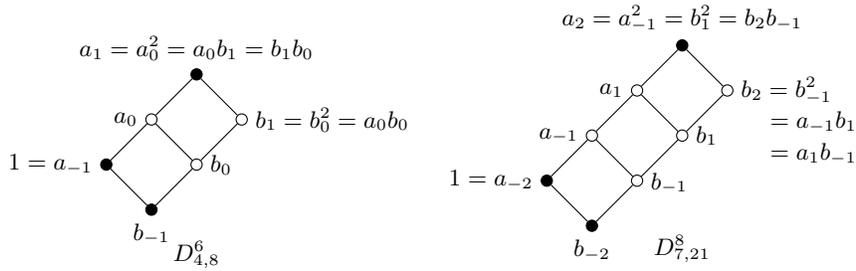
\begin{figure}
\begin{tikzpicture}
\begin{scope}[xshift=0cm,scale=0.6]
\node[draw,circle,inner sep=1.5pt,fill] (bot) at (-1,0) {};
\node[draw,circle,inner sep=1.5pt,fill] (1) at (-2,1) {};
\node[draw,circle,inner sep=1.5pt] (b) at (0,1) {};
\node[draw,circle,inner sep=1.5pt] (a) at (-1,2) {};
\node[draw,circle,inner sep=1.5pt] (0) at (1,2) {};
\node[draw,circle,inner sep=1.5pt,fill] (top) at (0,3) {};
\path [-] (bot) edge node {} (1);
\path [-] (bot) edge node {} (b);
\path [-] (1) edge node {} (a);
\path [-] (b) edge node {} (a);
\path [-] (b) edge node {} (0);
\path [-] (0) edge node {} (top);
\path [-] (a) edge node {} (top);   
\node[label,anchor=south,yshift=-1pt] at (top) {$a_1 = a_0^2=a_0b_1=b_1b_0$};
\node[label,anchor=east,xshift=1pt] at (a) {$a_0$};
\node[label,anchor=east,xshift=1pt] at (1) {$1=a_{-1}$};
\node[label,anchor=west,xshift=-1pt] at (0) {$b_1=b_0^2=a_0b_0$};
\node[label,anchor=west,xshift=-1pt] at (b) {$b_0$};
\node[label,anchor=north,yshift=1pt] at (bot) {$b_{-1}$};
\node[] at (0,-1) {$D^6_{4,8}$};
\end{scope} 
\end{tikzpicture}
\begin{tikzpicture}
\begin{scope}[xshift=0cm,scale=0.6]
\node[draw,circle,inner sep=1.5pt] (b-1) at (0,-1) {};
\node[draw,circle,inner sep=1.5pt,fill] (bot) at (-1,-2) {};
\node[draw,circle,inner sep=1.5pt,fill] (1) at (-2,-1) {};
\node[draw,circle,inner sep=1.5pt] (b1) at (1,0) {};
\node[draw,circle,inner sep=1.5pt] (0) at (2,1) {};
\node[draw,circle,inner sep=1.5pt] (a-1) at (-1,0) {};
\node[draw,circle,inner sep=1.5pt] (a1) at (0,1) {};
\node[draw,circle,inner sep=1.5pt,fill] (top) at (1,2) {};
\path [-] (bot) edge node {} (1);
\path [-] (bot) edge node {} (b-1);
\path [-] (1) edge node {} (a-1);
\path [-] (b-1) edge node {} (a-1);
\path [-] (b-1) edge node {} (b1);
\path [-] (b1) edge node {} (0);
\path [-] (b1) edge node {} (a1);
\path [-] (0) edge node {} (top);
\path [-] (a-1) edge node {} (a1);
\path [-] (a1) edge node {} (top);   
\node[label,anchor=south,yshift=-1pt] at (top) {$a_2 = a_{-1}^2=b_1^2=b_2b_{-1}$};
\node[label,anchor=east,xshift=1pt] at (a1) {$a_1$};
\node[label,anchor=east,xshift=1pt] at (a-1) {$a_{-1}$};
\node[label,anchor=east,xshift=1pt] at (1) {$1=a_{-2}$};
\node[label,anchor=west,xshift=-1pt] at (0) {$b_2=b_{-1}^2$};
\node[label,anchor=west,yshift=-12pt,xshift=10pt] at (0) {$=a_{-1}b_1$};
\node[label,anchor=west,yshift=-24pt,xshift=10pt] at (0) {$=a_1b_{-1}$};
\node[label,anchor=west,xshift=-1pt] at (b1) {$b_1$};
\node[label,anchor=west,xshift=-1pt,yshift=-1pt] at (b-1) {$b_{-1}$};
\node[label,anchor=north,yshift=1pt] at (bot) {$b_{-2}$};
\node[] at (1,-2.5) {$D^8_{7,21}$};
\end{scope}
\end{tikzpicture}\\
\caption{The DqRAs $\mathbf{A}_3$ and $\mathbf{A}_4$. They have the identifiers $D^6_{4,8}$ and $D^8_{7,21}$
(see~\cite{CJR-DqRAs}). The black nodes denote idempotent elements.}\label{fig:An_for_n=3,4}
\end{figure}
\end{example}

\subsection{Involutive Partially Ordered Monoids}\label{sub:Ipo-monoid}

A \emph{partially ordered monoid} \normalfont{(}\emph{pomonoid}\normalfont{)} is a structure of the form $\mathbf{P} = \langle P, \leqslant, \cdot, 1\rangle$ such that $\langle P, \leqslant\rangle$ is a poset, $\langle P, \cdot, 1\rangle$ is a monoid, and for all $x, y \in P$, if $x \leqslant y$, then $x\cdot z \leqslant y \cdot z$ and $z\cdot x \leqslant z \cdot y$.  
Involutive partially ordered monoids have their origins in the RAMiCS community. They were introduced by Gil-Ferez et al.~\cite{GJL23} and also studied by Bonzio et al.~\cite{BGJPS24}. 


\begin{definition}{\normalfont\cite{GJL23}}\label{def:Ipo-monoid}
An \emph{involutive partially ordered monoid} \normalfont{(}or \emph{ipo-monoid}\normalfont{)} is a structure of the form $\mathbf{P} = \langle P, \leqslant, \cdot, 1, ^-, ^{\sim}\rangle$ such that $\langle P, \leqslant\rangle $ is a poset, $\langle P, \cdot, 1\rangle$ is a monoid, and the following holds for all $x, y \in P$:
\begin{equation}
x\leqslant y \quad \textnormal{ iff } \quad  x\cdot y^{\sim} \leqslant 1^- \quad \textnormal{ iff } \quad y^-\cdot x \leqslant 1^-.\label{eq:ipo}
\end{equation}
\end{definition}

We will denote $1^-$ by $0$. The unary operations $^\sim$ and $^-$ are called \emph{involutive
negations}. Given an involutive partially monoid $\mathbf{P}= \langle P, \leqslant, \cdot, 1, ^-, ^{\sim}\rangle$, we say that $\mathbf{P}$  is \emph{cyclic} if $x^\sim = x^-$ for all $x \in P$. We say $\mathbf{P}$ is \emph{commutative}
if $x \cdot y = y\cdot x$ for all $x, y \in P$. It is easy to see that if $\mathbf{P}$ is commutative, then it is cyclic.  
As mentioned before, we will sometimes omit  $\cdot$ and
write $xy$ 
for simplicity. 

The properties of the next lemma
will often be used without
referring to them explicitly.

\begin{lemma}{\normalfont\cite{GJL23}}\label{lem:properties_ipo-mon}
%
Let $\mathbf{P} = \langle P, \leqslant, \cdot, 1, ^-, ^{\sim}\rangle$ be an ipo-monoid. Then the following hold for all $x, y, z\in P$:
\begin{enumerate}[\normalfont (i)]
\item $x^{-{\sim}} = x = x^{{\sim}{-}}$,
\item $x \leqslant y$ iff $y^- \leqslant x^-$ iff $y^{\sim}\leqslant x^{\sim}$,
\item $0 = 1^{\sim}$, $0^{\sim} = 1$, $0^{-} = 1$, 
\item $xy \leqslant z$ iff $x \leqslant (y  z^{\sim})^-$ iff $y \leqslant (z^-x)^{\sim}$, and
\item $x \leqslant y$ then $xz \leqslant yz$ and $zx\leqslant zy$.
\end{enumerate}
\end{lemma}

The following proposition gives an alternative presentation of ipo-monoids. We remark that when using Prover9/Mace4~\cite{P9M4}, the axiomatization below is faster when generating finite models. 

\begin{proposition}\label{prop:alternative_def_ipo}
A structure $\mathbf{P} = \langle P, \leqslant, \cdot, 1, ^-, ^{\sim}\rangle$ is an ipo-monoid iff $\langle P, \leqslant, \cdot, 1\rangle$ is a pomonoid and the following conditions hold for all $x,y, z \in P$:
\begin{multicols}{2}
\begin{enumerate}[\normalfont (i)]
\item $x^{-\sim} \leqslant x$ and $x^{\sim -} \leqslant x$, and 
\item $x y \leqslant z^{\sim}$ iff $zx\leqslant y^{-}$.
\end{enumerate}
\end{multicols}
\end{proposition}

\begin{proof}
Let $\mathbf{P} = \langle P, \leqslant, \cdot, 1, ^{\sim}, ^-\rangle$ be an ipo-monoid. It follows from Lemma~\ref{lem:properties_ipo-mon}(v) that the monoid operation is order preserving in both coordinates. Item (i) above follows from Lemma~\ref{lem:properties_ipo-mon}(i). We thus have to prove that item (ii) above holds. Assume $x y \leqslant z^{\sim}$. Then $z x y \leqslant z  z^{\sim}$, and so, since $z \leqslant z$ implies $z z^{\sim} \leqslant 0$, we have $z x y \leqslant 0$. Hence, $z  x y^{-\sim} \leqslant 0$, and therefore $z x \leqslant y^{-}$ by (\ref{eq:ipo}). Conversely, assume $z x \leqslant y^{-}$. Then $z x y^{-\sim}\leqslant 0$ by (\ref{eq:ipo}), and so $z x y \leqslant 0$. This is equivalent to $z^{\sim -}xy \leqslant 0$, and thus $xy \leqslant z^{\sim}$. 

Assume $\langle P, \leqslant, \cdot, 1\rangle$ is a pomonoid and (i) and (ii) hold. We must show that (\ref{eq:ipo}) holds. We first show that $x \leqslant x^{\sim -}$ and $x \leqslant x^{-\sim}$ for all $x \in P$. Let $x \in P$. Then $1\cdot x^{\sim}\leqslant x^\sim$, and so $x \cdot 1 \leqslant x^{\sim -}$ by (ii), which implies $x \leqslant x^{\sim -}$. Similarly, $x^{-} \cdot 1 \leqslant x^{-}$, so $1\cdot  x \leqslant x^{-\sim}$, and consequently $x \leqslant x^{-\sim}$. It thus follows that $x = x^{-\sim}$ and $x =x^{\sim -}$ for all $x \in P$. 

Next we show that $1^- = 1^\sim$. Since $1 \cdot 1^{-} \leqslant 1^{-}$, we have $1^{-}\cdot 1 \leqslant 1^{\sim}$ by (ii), which means $1^- \leqslant 1^\sim$. Similarly, $1^\sim \cdot 1 \leqslant 1^\sim$ implies $1 \cdot 1^\sim \leqslant 1^{-}$, so $1^\sim \leqslant 1^-$. Hence, $1^{-} = 1^{\sim}$. 

Now assume $x \leqslant y$. Then $1 \cdot x \leqslant y^{\sim -}$, and so, by (ii), $x y^\sim \leqslant 1^\sim = 1^{-}$. Conversely, assume $x y^\sim \leqslant  1^{\sim}$. Then $1\cdot  x \leqslant y^{\sim -}$. Hence, by (i), $1 \cdot x \leqslant y$, i.e., $x \leqslant y$. Similarly, $x \leqslant y$ iff $x \cdot 1 \leqslant y^{-\sim}$ iff $y^{-}x \leqslant 1^-$.
\qed
\end{proof}

\subsection{Pregroups}\label{sub:pregroups}

Pregroups were first described by Lambek~\cite{Lam99}. They are special ipo-monoids. 

\begin{definition}\label{def:pregroup}
A \emph{pregroup} is a structure $\mathbf{P} = \langle P, \leqslant, \cdot, 1, ^\ell, ^r\rangle$ such that $\langle P, \leqslant, \cdot, 1\rangle$ is a partially ordered monoid and the  following conditions hold for all $x \in P$:
\begin{equation}
x^\ell x \leqslant 1 \leqslant x x^\ell \qquad \textnormal{ and } \qquad x x^r \leqslant 1 \leqslant x^r  x.
\end{equation}
\end{definition}
The elements $x^{\ell}$ and $x^r$ are called the \emph{left adjoint} and the \emph{right
adjoint}, respectively, of $x$. If $x^\ell = x^r$, then we say the pregroup is \emph{cyclic}.

\begin{example}\label{ex:group_is_pregroup}
If $\mathbf{G} = \langle G, \cdot, e, ^{-1}\rangle$ is a group, then $\langle G, =, \cdot, e, ^{-1}, ^{-1}\rangle$ is a pregroup. Every finite pregroup is of this form. 
\end{example}

The following is an example of a non-cyclic pregroup and can be found in~\cite{Lam99}.

\begin{example}\label{ex:Lambek pregroup}
Consider the set $\textsf{U}(\mathbb{Z})$ consisting of all unbounded, monotone  functions $f: \mathbb{Z} \to \mathbb{Z}$. For all $f, g \in \textsf{U}\left(\mathbb{Z}\right)$, define $f \leqslant_{\textsf{U}\left(\mathbb{Z}\right)} g$ iff $f(x) \leqslant_{\mathbb{Z}} g(x)$ for all $x \in \mathbb{Z}$. Further, for each unbounded monotone function $f$, define $f^{\ell}(x) = \textnormal{min}\{y \in \mathbb{Z} \mid x \leqslant f(y)\}$ and $f^r(x) = \textnormal{max}\{y \in \mathbb{Z} \mid f(y) \leqslant x\}$ for all $x \in \mathbb{Z}$. 
Then $\langle \textsf{U}(\mathbb{Z}), \leqslant_{\textsf{U}\left(\mathbb{Z}\right)}, \circ, \textnormal{id}_\mathbb{Z}, ^\ell, ^r \rangle$ is a pregroup, where $\textnormal{id}_\mathbb{Z}$ is the identity map. 
\end{example}

The properties in the following proposition are easy to prove. 

\begin{proposition}\label{prop:Properties_of_pregroups}
Let $\mathbf{P} = \langle P, \leqslant, \cdot, 1, ^\ell, ^r\rangle$ be a pregroup. Then: 
\begin{enumerate}[\normalfont (i)]
\item $x^{\ell r} = x = x^{r\ell}$,
\item $1^\ell = 1 = 1^r$,
\item $\left(x y\right)^\ell = y^\ell x^\ell$,
\item $\left(x y\right)^r = y^r x^r$, and 
\item $x \leqslant y$ iff $y^{\ell} \leqslant x^{\ell}$ iff $y^r \leqslant x^r$. 
\end{enumerate}
\end{proposition}

In the next proposition we will show that any pregroup is an ipo-monoid. 

\begin{proposition}\label{prop:pregroups_are_ipo-monoids}
Every pregroup is an ipo-monoid. 
\end{proposition}

\begin{proof}
Let $\mathbf{P} = \langle P, \leqslant, \cdot, 1, ^\ell, ^r\rangle$ be a pregroup. 
We have to show that $x \leqslant y$ iff $x y^r \leqslant 1^\ell = 1$ iff $y^\ell  x \leqslant 1$. Assume $x\leqslant y$. Then $x y^r \leqslant y y^r$, and so, since $y y^r \leqslant 1$, we get $x y^r \leqslant 1$. For the converse, assume $x y^r \leqslant 1$. Then $x  y^r y \leqslant 1\cdot y = y$. Now $1 \leqslant y^r y$, so $x = x \cdot 1 \leqslant x  y^r  y$. Therefore, $x \leqslant y$. The other equivalence can be proved in a similar way. 
\qed
\end{proof}

\section{Distributive residuated lattices of binary relations}\label{sec:embedDRL}

The following example of a concrete distributive residuated lattice comes from the PhD thesis of Galatos~\cite{G03}. For a pomonoid  $\mathbf{P}=\langle P,\leqslant,\cdot,1\rangle$, we write $\mathsf{Up}(P,\leqslant )$ for the set of up-sets of $\langle P,\leqslant\rangle $. For $U,V \subseteq P$, let 
$U \bullet V = {\uparrow}\left\{\,x y \mid x \in U , y \in V\,\right\}$ and 
$U/V=\{\, z \mid \{z\}\bullet V \subseteq U \,\}$ and $V\backslash U = \{\, z \mid V\bullet \{z\} \subseteq U \,\}$.

\begin{proposition}{\normalfont\cite[Example 3.17]{G03}}\label{prop:RLUPP}
Let $\mathbf{P}=\langle P, \leqslant, \cdot, 1\rangle$ be a pomonoid. Then $\langle \mathsf{Up}(\mathbf{P}), \cap, \cup, \bullet, {\uparrow}1,/,\backslash\rangle$ is a distributive residuated lattice. 
\end{proposition}


Now we show that under certain conditions, when $\mathbf{P}=\langle P, \leqslant, \cdot, 1\rangle$ is a partially ordered monoid, the residuated lattice $\langle \mathsf{Up}(P,\leqslant) ,\cap, \cup, \bullet, {\uparrow}1, \backslash,/\rangle$ 
can be embedded into a distributive residuated lattice of binary relations.

First, we set up notation and recall some properties of binary relations. Let $X$ be a set and $R\subseteq X^2$. The converse  of $R$ is 
$R^\smile = \left\{\left(x, y\right) \mid \left(y, x\right) \in R\right\}$, and its complement is
$R^c=\{\, (x,y) \in X^2 \mid (x,y) \notin R\,\}$.
The identity relation is denoted by $\mathrm{id}_X=\{\,(x,x)\mid x \in X\,\}$.
For $R\subseteq X^2$ we have $(R^{\smile})^{\smile}=R$, $(R^{\smile})^c=(R^c)^{\smile}$, and $\mathrm{id}_X\mathbin{;} R = R\mathbin{;} \mathrm{id}_X = R$.
The composition of two binary relations $S$ and $R$
is given by 
$R \mathbin{;} S = \left\{\left(x, y\right) \mid \left(\exists z \in X\right)\left(\left(x, z\right) \in R \textnormal{ and } \left(z, y\right) \in S\right)\right\}$.
For $R,S, T\subseteq X^2$ we have $\left(R\, ; S\right)\mathbin{;} T = R\mathbin{;} \left(S\mathbin{;} T\right)$ and 
$\left(R\mathbin{;} S\right)^\smile = S^\smile\mathbin{;} R^\smile$.

Now, for a partially ordered monoid $\mathbf{P}=\langle P, \leqslant, \cdot, 1\rangle$, consider 
the twisted product order $\preccurlyeq$ on $P\times P$: $(x,y) \preccurlyeq (u,v)$ iff $u \leqslant x$ and $y \leqslant v$. 
The poset $(P^2,\preccurlyeq)$ gives rise to a distributive residuated lattice $ \langle \mathsf{Up}(
P^2,\preccurlyeq), \cap,\cup, ;, \leqslant, \backslash,/
\rangle$ where $R\backslash S=(R^\smallsmile ; S^c)^c$ and $R/S =(R^c;S^\smallsmile)^c$ (cf.~\cite[Section 3]{RDqRA25}). 

Now define a map $\sigma : \mathsf{Up}(P,\leqslant) \to \mathsf{Up}(P^2,\preccurlyeq)$ as follows:
$$\sigma(U)= {\uparrow}\{\, (g, g u) \mid g \in P, u \in U\, \}.$$

The map $\sigma$ is inspired by the Cayley map used for representing group relation algebras. 
The lemma below is important and will be used frequently in proofs later on. It is easy to prove using the order preservation of the monoid operation. 
\begin{lemma}\label{lem:alt_sigma}
Let $\mathbf{P}=\langle P,\leqslant,\cdot, 1\rangle$ be a pomonoid with $u \in P$ and $U \in \mathsf{Up}(P,\leqslant)$.   
\begin{enumerate}[\normalfont (i)]
\item $(x,y) \in \sigma(U)$ iff there exists $u \in U$ such that  $x u \leqslant y$.
\item $(x,y) \in \sigma({\uparrow}u)$ iff $x u \leqslant y$.
\end{enumerate}
\end{lemma}

In general, $\sigma$ will not be a residuated lattice embedding, but it will satisfy many of the required conditions. 
\begin{lemma}
If $\mathbf{P}$
is a pomonoid, then $\sigma : \mathsf{Up}(P,\leqslant) \to \mathsf{Up}(P^2,\preccurlyeq)$ is one-to-one, and it preserves arbitrary joins, the fusion operation $\bullet$, and the monoid identity.   
\end{lemma}
\begin{proof}
Suppose $U \neq V$ and without loss of generality, consider  $x \in U$ such that $x \notin V$. Clearly  $(1,x) \in \sigma(U)$. Suppose that $(1,x) \in \sigma(V)$. 
By Lemma~\ref{lem:alt_sigma}(i) there exists $v \in V$ with $1\cdot v \leqslant x$. Since $V$ an upset, we get $x \in V$, a contradiction. 



It is easy to see that $\sigma$ is order-preserving, hence $\bigcup_{i\in I} \left( \sigma(U_i)\right) \subseteq \sigma \left( \bigcup_{i\in I} U_i\right) $. If $(x,y) \in \sigma \left( \bigcup_{i\in I} U_i \right)$ then there exists $u \in \bigcup_{i \in I} U_i$ and $g \in P$ such that $(g,gu)\preccurlyeq (x,y)$. Then $u \in U_j$ for some $j \in I$ and so  $(x,y) \in \sigma(U_j) \subseteq \bigcup_{i\in I} \left( \sigma(U_i)\right)$.

Now we show that $\sigma$ preserves the monoid operation. First, observe that if $U=\varnothing$ or $V = \varnothing$, then we trivially get $\sigma(U \bullet V) =\sigma(U) \mathbin{;}  \sigma (V)$. If $(x,y) \in \sigma (U \bullet V)$ then by Lemma~\ref{lem:alt_sigma}(i) there exists $t \in U\bullet V$ such that $xt \leqslant y$. Now $uv \leqslant t$ for some $u \in U$ and $v \in V$ which implies $xuv \leqslant xt \leqslant y$.  
From $u \in U$ we get $(x,xu) \in \sigma(U)$. Since $v \in V$ and $xuv \leqslant y$ we get $(xu,y) \in \sigma(V)$. Hence $(x,y) \in \sigma(U)\mathbin{;}\sigma(V)$.

Now let $(x,y) \in \sigma(U) ; \sigma(V)$. There exists $z$ with $(x,z)\in \sigma(U)$ and $(z,y) \in \sigma(V)$. By Lemma~\ref{lem:alt_sigma}(i) there exist $u \in U$ and $v \in V$ such that $xu \leqslant z$ and $zv\leqslant y$. Hence $x(uv) \leqslant zv \leqslant y$ and since $uv \in U \bullet V$ we have $(x,y) \in \sigma(U \bullet V)$.

By Lemma~\ref{lem:alt_sigma}(ii) we get $\sigma({\uparrow}1)=\{\,(x,y) \mid x\cdot 1 \leqslant y \,\}= {\leqslant}$, which shows that $\sigma$ preserves the monoid identity. 
\qed

\end{proof}

The map $\sigma$ does \emph{not} always preserve meets if $P$ is a pomonoid. Consider the discretely ordered pomonoid with set $\{0,1\}$ and $0^2=0$. We get $\sigma(\{0\})=\{(0,0),(1,0)\}$ and $\sigma(\{1\})=\{(0,0),(1,1)\}$. Now $\sigma(\{0\} \cap \{1\})=\sigma(\varnothing)=\varnothing$. However $\sigma(\{0\}) \cap \sigma(\{1\})=\{(0,0)\}$.  

Below we characterise those pomonoids for which $\sigma$ will preserve meets and hence will be a (distributive) residuated lattice embedding. 


\begin{proposition}\label{prop:CondW}
Let $\langle P,\leqslant,\cdot,1 \rangle$ be a pomonoid. The map $\sigma$ preserves meets if and only if  whenever $x  u \leqslant y $ and $x v \leqslant y$, there exists $w$ such that $u\leqslant w$, $v \leqslant w $ and $x  w \leqslant y$.     \end{proposition}
\begin{proof}
Assume $\sigma$ preserves meets and that $xu \leqslant y$ and $xv \leqslant y$. By Lemma~\ref{lem:alt_sigma}(ii) we have 
$(x,y) \in \sigma({\uparrow}u)\cap \sigma({\uparrow}v)$. Then $(x,y) \in  \sigma({\uparrow}u \cap {\uparrow}v)$ and so by Lemma~\ref{lem:alt_sigma}(i) there exists $w \in {\uparrow}u \cap {\uparrow}v$ such that $xw \leqslant y$. 

Now assume that $\mathbf{P}$ satisfies Condition W. 
Since $\sigma$ preserves joins, it is order preserving, so $\sigma(U \cap V) \subseteq \sigma(U) \cap \sigma(V)$. Let $(x,y) \in \sigma(U) \cap \sigma (V)$. Then there exist $g, h \in P$, $u \in U$, $v \in V$ such that $(g,gu) \preccurlyeq (x,y)$ and $(h,hv) \preccurlyeq (x,y)$. Since $x \leqslant g$ and $x \leqslant h$ we have $xu\leqslant gu \leqslant y$ and $xv \leqslant hu \leqslant y$. Our assumption gives us  $w \in U \cap V$ with $xw \leqslant y$. Hence $(x,xw) \in \sigma (U \cap V)$ and $(x,xw) \preccurlyeq (x,y)$.
\qed
\end{proof}
We will refer to the condition in Proposition~\ref{prop:CondW} as \emph{Condition W}. 
There are four important examples of pomonoids for which Condition W holds. 

\begin{theorem}\label{thm:sigma-pres-meets}
The map $\sigma$ is a (distributive) residuated lattice embedding if $\langle P, \leqslant, \cdot,1 \rangle$ is the pomonoid reduct of 
{\upshape(i)} a residuated pomonoid, {\upshape(ii)} an ipo-monoid, {\upshape(iii)} a lattice-ordered monoid, or {\upshape(iv)} a pregroup. 
\end{theorem}
\begin{proof} 
The fact that $\sigma$ is one-to-one and preserves joins, fusion, and the monoid identity follows from the results above. We show that Condition W is satisfied for  (i)--(iii). Assume that $xu \leqslant y$ and $xv \leqslant y$.

(i) If $\langle P, \leqslant,\cdot, 1\rangle$ is a residuated pomonoid,  let $w = x\backslash y$.  (ii) From Lemma~\ref{lem:properties_ipo-mon}(iv), an ipo-monoid is residuated. 
(iii) If $\langle P, \leqslant,\cdot,1\rangle$ is a lattice-ordered monoid then $w=u \vee v$ meets the requirements since fusion distributes over joins. 

(iv) We use Proposition~\ref{prop:pregroups_are_ipo-monoids} and item (ii).  
\qed

\end{proof}

The main consequence of the theorem above is that if a distributive residuated lattice is isomorphic to a subalgebra of the algebra of up-sets of a residuated pomonoid, a lattice-ordered monoid, or a pregroup, then it is representable as an algebra of binary relations with fusion given by relational composition. 

\section{Representable DInFL-algebras and DqRAs}\label{sec:representable}

In recent papers~\cite{RDqRA25,CR24}, we gave a definition of what is means for a DInFL-algebra or DqRA to be representable. The definition relies on building algebras from posets, similar to the way relation algebras are built from sets. 

The operations defined on binary relations defined at the start of Section~\ref{sec:embedDRL} will be used throughout this section. 

The following lemma (a simplified version of ~\cite[Lemma 3.4]{RDqRA25}) will be used frequently.
Importantly, it can be applied if  $\gamma$ is a bijective function $\gamma : X \to X$.
\begin{lemma}{\normalfont \cite[Lemma 3.4]{RDqRA25}}
\label{Lemma:ComplementComposition}
Let $X$ be a set and $R, S, \gamma \subseteq X^2$. If 
$\gamma$
satisfies 
$\gamma^{\smile}\mathbin{;} \gamma = \mathrm{id}_X$ and $\gamma \mathbin{;} \gamma^{\smile}=\mathrm{id}_X$
then $\left(\gamma\mathbin{;} R\right)^c = \gamma \mathbin{;} R^c$ and $\left(R\mathbin{;} \gamma\right)^c = R^c \mathbin{;} \gamma$. 
\end{lemma}

As in Section~\ref{sec:embedDRL}, for a poset  $\mathbf X = \langle X, \leqslant\rangle$, we will consider the distributive residuated lattice 
$\langle \mathsf{Up}(X^2,\preccurlyeq), \cap, \cup, \mathbin{;},\backslash,/, \leqslant \rangle$. 

We will now define how an  order automorphism $\alpha : X \to X $ and a self-inverse dual order automorphism $\beta : X \to X $ are then used to define the unary operations ${\sim}$, $-$ and $\neg$. Details can be found in~\cite[Section 3]{RDqRA25}.


\begin{theorem}{\normalfont \cite[Theorems 3.12 and 3.15]{RDqRA25}}\label{Theorem:Dq(E)}
Let $\mathbf{X}=\langle X,\leqslant\rangle$ be a poset and $\alpha: X \to X$ an order automorphism of $\mathbf X$. Set $1={\leqslant}$ and for 
$R \in \mathsf{Up}(X^2,\preccurlyeq)$, define
${\sim} R = R^{c\smile}\mathbin{;} \alpha$, $-R = \alpha \mathbin{;} R^{c\smile}$. Then 
$$ \langle \mathsf{Up}(X^2,\preccurlyeq), \cap, \cup, \mathbin{;},1, -, {\sim}\rangle 
$$
is a DInFL-algebra and it is cyclic iff $\alpha$ is the identity. Further, if $\beta: X \to X$ is self-inverse dual order automorphism of $\mathbf X$ such that  $\beta = \alpha \mathbin{;} \beta\mathbin{;} \alpha$, then defining 
$\neg R =  \alpha\mathbin{;} \beta \mathbin{;} R^c \mathbin{;} \beta$ we get 
$$ 
\left\langle \mathsf{Up}(X^2,\preccurlyeq),\cap, \cup, \mathbin{;}, 1, -, {\sim}, {\neg} \right\rangle$$
is a distributive quasi relation algebra. 
\end{theorem}	

Algebras of the form described by Theorem~\ref{Theorem:Dq(E)} are called 
\emph{full} DInFL-algebras or \emph{full} DqRAs 
and the classes of such algebras are denoted by $\mathsf{FDInFL}$ and $\mathsf{FDqRA}$, respectively. 
The definition below appeared first for DqRAs~\cite[Definition 4.5]{RDqRA25} and then later for DInFL-algebras~\cite[Definition 2.14]{CR24}. 


\begin{definition}
{\normalfont\cite{RDqRA25,CR24}}\label{Definition:RDqRA}
A DInFL-algebra $\mathbf{A}=\left\langle A, \wedge, \vee, \cdot, 1, -, {\sim}, \right\rangle$ will be  called \emph{representable} if $\mathbf{A}\in \mathbb{ISP}(\mathsf{FDInFL})$ and a 
DqRA $\mathbf{B} = \left\langle B, \wedge, \vee, \cdot, 1, -, {\sim}, {\neg}\right\rangle$ 
is \emph{representable} if 
$\mathbf{B} \in \mathbb{ISP}\left(\mathsf{FDqRA}\right)$.
\end{definition}

A DqRA $\mathbf{A}$ is \emph{finitely} representable if the poset $\langle X ,\leqslant\rangle$ used in the representation of $\mathbf{A}$ is finite.

\begin{remark}
As for the definition of representable relation algebras, the definition of representable DInFL-algebra and representable DqRA  can be stated in an equivalent way, but where $X^2$ is replaced with an arbitrary equivalence relation $E$ satisfying certain conditions. The setting described above is simpler and still sufficient for our purposes in this paper.     
\end{remark}


\section{Representability results for DInFL-algebras}\label{sec:Rep_DInFL}

Now we  construct a DInFL-algebra from an ipo-monoid. 
Let $\mathbf{P} = \langle P, \leqslant, \cdot, 1, ^-, ^{\sim}\rangle$ be an ipo-monoid. For each $U \in \mathsf{Up}\left(P, \leqslant\right)$, define 
$-U = \left\{x^- \mid x \notin U\right\}$ and ${\sim}U = \left\{x^{\sim}\mid x \notin U\right\}$. Below 
we show that $\textsf{Up}\left(P, \leqslant\right)$ is closed under $\sim$ and $-$. 

\begin{lemma}\label{lem:-U&~U_upsets}
Let $\mathbf{P} = \langle P, \leqslant, \cdot, 1, ^-, ^{\sim}\rangle$ be an ipo-monoid. If $U \in \mathsf{Up}\left(P, \leqslant\right)$, then ${\sim}U \in \mathsf{Up}\left(P, \leqslant\right)$ and $-U \in \mathsf{Up}\left(P, \leqslant\right)$. 
\end{lemma}

\begin{proof}
Let $U \in \mathsf{Up}\left(P, \leqslant\right)$. To see that ${\sim}U \in  \mathsf{Up} \left(P, \leqslant\right)$, let $x\in {\sim}U$ and $y \in P$, and assume $x \leqslant y$. Since $x \in {\sim}U$, it follows from item  (i) of Lemma~\ref{lem:properties_ipo-mon} that $x^{-{\sim}} \in {\sim}U$, and hence $x^{-} \notin U$. Since $x\leqslant y$, we have $y^- \leqslant x^-$ by item (ii) of Lemma~\ref{lem:properties_ipo-mon}. Hence, since $U$ is an upset of $\langle P, \leqslant\rangle$, it follows that $y^- \notin U$. This shows that $y^{-\sim} \in {\sim}U$. Consequently, $y\in {\sim}U$. In a similar way we can show that $-U$ is an upset of $\langle P, \leqslant\rangle$. 
\qed
\end{proof}

Recall that we define $U \bullet V = {\uparrow}\left\{\,x y \mid x \in U , y \in V\,\right\}$ for all $U, V \subseteq P$. Since $\left\langle \mathsf{Up}\left(P, \leqslant\right), \cap, \cup\right\rangle$ is a distributive lattice and  $\left\langle\mathsf{Up}\left(P, \leqslant\right),  \bullet, {\uparrow}1\right\rangle$ is a monoid by Proposition~\ref{prop:RLUPP}, it will follow that $\left\langle \mathsf{Up}\left(P, \leqslant\right), \cap, \cup, \bullet, {\uparrow}1, -, \sim\right\rangle$ is a DInFL-algebra if we can show that that $\sim$ and $-$ satisfy (\ref{eq:InFL}).

\begin{theorem}\label{thm:DInFL_from_ipo-monoid}
Let $\mathbf{P} = \langle P, \leqslant, \cdot, 1, ^-, ^{\sim}\rangle$ be an ipo-monoid. Then the structure $\mathcal{D}(\mathbf{P})=\left\langle \mathsf{Up}\left(P, \leqslant\right), \cap, \cup, \bullet, {\uparrow}1, -, \sim\right\rangle$ is a DInFL-algebra. Moreover, the algebra $\mathcal{D}(\mathbf{P})$
is cyclic iff $\mathbf{P}$ is cyclic. 
\end{theorem}

\begin{proof}
We have to show that $T\bullet U \subseteq V$ iff $U \subseteq {\sim}\left(-V \bullet T\right)$ iff $T \subseteq -\left(U\bullet {\sim}V\right)$ for all $T, U, V \in  \mathsf{Up}\left(P, \leqslant\right)$. First, assume $T\bullet U \subseteq V$ and let $y \in U$. Suppose 
that $y \notin {\sim}\left(-V \bullet T\right)$. Then $y^- \in -V \bullet T$. Hence, there is some $z \in -V$ and $x \in T$ such that $zx \leqslant y^-$. It follows that $zxy  \leqslant y^- y$. Now $y \leqslant y$, so $y^- y \leqslant 0$ by (1), and therefore we obtain $z xy  \leqslant 0$. 
This is equivalent to $z^{\sim -}x y \leqslant 0$, and therefore, by (1), $xy \leqslant z^{\sim}$. 
Since $x \in T$ and  $y \in U$, we get $z^{\sim} \in T\bullet U \subseteq V$. This shows that $z^{\sim -} \notin -V$, which implies $z \notin -V$, a contradiction. 

Now assume $U \subseteq {\sim}\left(-V \bullet T\right)$ and let $z \in T \bullet U$. Suppose $z \notin V$. Then $z^- \in - V$. Since $z \in T \bullet U$, there exist $x \in T$ and $y \in U$ such that $xy \leqslant z$. Hence, $z^- x y \leqslant z^-z$, and so, since $z \leqslant z$ implies $z^- z \leqslant 0$, we get $z^- x y \leqslant 0$.  This is equivalent to $z^- x y^{-\sim} \leqslant 0$, and consequently $z^- x \leqslant y^-$. Now $z^- \in -V$ and $x \in T$, so $y^{-} \in -V\bullet T$, which means $y = y^{-\sim} \notin {\sim}\left(-V \bullet T\right)$. 
Thus, since $U \subseteq {\sim}\left(-V \bullet T\right)$, we get $y \notin U$, which is a contradiction. 

Next assume $T\bullet U \subseteq V$ and let $x \in T$. Suppose that $x \notin {-}\left(U \bullet {\sim} V\right)$. This implies $x^{\sim} \in U \bullet {\sim} V$. Hence, there is some $y \in U$ and $z \in {\sim} V$ such that $y z \leqslant x^{\sim}$. It follows that $x y z \leqslant x x^{\sim}$, and therefore, since $x \leqslant x$ implies $x  x^{\sim} \leqslant 0$, we have $x y z \leqslant 0$. This is equivalent to 
$x y z^{-\sim} \leqslant 0$, and therefore $x y \leqslant z^-$. 
Hence, since $x \in T$ and $y \in U$, we get $z^- \in T\bullet U\subseteq V$, which means $z = z^{- \sim} \notin {\sim}V$, which is a contradiction. 

For the converse implication, assume $T \subseteq -\left(U\bullet {\sim}V\right)$ and let $z \in T\bullet U$. Suppose $z \notin V$. Since $z\in T\bullet U$, there exist $x \in T$ and $y \in U$ such that $xy \leqslant z$. Hence, we get $x y z^{\sim} \leqslant z  z^{\sim}$, and so, since $z z^{\sim} \leqslant 0$, we obtain $x y z^{\sim} \leqslant 0$. This is equivalent to $x^{\sim -} y z^{\sim} \leqslant 0$, and consequently $y z^{-} \leqslant x^{\sim}$ by (1). 
Now recall that $z \notin V$, so $z^\sim \in {\sim}V$, and therefore, since it is also the case that $y \in U$, we get $x^\sim \in U\bullet {\sim}V$. This shows that $x = x^{\sim -} \notin -\left(U \bullet {\sim}V\right)$. But $T \subseteq -\left(U\bullet {\sim}V\right)$, so $x \notin T$, which is a contradiction. 

It is straightforward to show that if $x^{\sim} = x^-$ for all $x \in P$, then the algebra $\left\langle \mathsf{Up}\left(P, \leqslant\right), \cap, \cup, \bullet, {\uparrow}1, -, \sim\right\rangle$ is cyclic. For the converse implication, suppose there is some $x \in P$ such that $x^\sim \neq x^-$. Then $x^\sim \not\leqslant x^-$ or $x^- \not\leqslant x^{\sim}$. Assume the first. This means $x^{\sim} \in {\uparrow}x^{\sim}$ and $x^- \notin {\uparrow}x^{\sim}$,
i.e. 
$x^{\sim -} \notin - {\uparrow}x^{\sim}$ and $x^{- \sim} \in {\sim}{\uparrow}x^{\sim}$. Hence $\mathcal{D}(\mathbf{P})$ is not cyclic. A similar proof applies if $x^- \nleqslant x^{\sim}$. 
\qed
\end{proof}

We say that a DInFL-algebra $\mathbf{A}$ is  a \emph{pregroup DInFL-algebra} if it is isomorphic to a subalgebra  of the algebra $\mathcal{D}(\mathbf{P})$  for some pregroup $\mathbf{P} = \langle P, \leqslant, \cdot, 1, ^{\ell}, ^r\rangle$. 

It is also possible to construct a DInFL-algebra of binary relations using an ipo-monoid.
Fix an ipo-monoid $\mathbf{P} = \langle P, \leqslant, \cdot, 1, ^{-}, ^{\sim}\rangle$ and define $\alpha: P \to P$ by setting, for all $x \in P$, $\alpha(x) = x^{\sim\sim}$. 
It follows from Lemma~\ref{lem:properties_ipo-mon} 
that 
%
 $\alpha$ is an order automorphism of $\langle  P, \leqslant\rangle $. 
%
%
Hence by Theorem~\ref{Theorem:Dq(E)} we obtain the following result.

\begin{theorem}\label{thm:DInFL_from_PtimesP}
Let $\mathbf P = \langle P, \leqslant, \cdot, 1, ^-, ^{\sim}\rangle$ be an ipo-monoid. Then the algebra of binary relations $\left\langle \mathsf{Up}\left(P^2, \preccurlyeq\right), \cap, \cup, \mathbin{;}, \leqslant, -, \sim\right\rangle$ is a DInFL-algebra. 
\end{theorem}

We will now show that pregroup DInFL-algebras are representable. We do this by 
proving  
that the DInFL-algebra $\mathcal{D}(\mathbf{P})$ 
constructed from a pregroup $\mathbf{P}$
can be embedded into the algebra  $\left\langle \mathsf{Up}\left(P^2, \preccurlyeq\right), \cap, \cup, \mathbin{;}, \leqslant, -, \sim\right\rangle$.  Consider the map $\sigma : \mathsf{Up}(P,\leqslant) \to \mathsf{Up}(P^2,\preccurlyeq)$ defined by $\sigma(U)= {\uparrow}\{\, (g, g u) \mid g \in P, u \in U\, \}$ for all $U\in \mathsf{Up}(P,\leqslant)$. In Section~\ref{sec:embedDRL} it was shown that $\sigma$ is injective and that it preserves meets, joins, the monoid operation and its identity. We therefore just have to check if $\sigma$ preserves the unary operations $\sim$ and $-$. To show that $\sigma$ preserves $\sim$ we will need the fact that 
$\alpha^{-1}(x) = x^{\ell\ell}$ for all $x \in P$.  This follows immediately from the definition of $\alpha$ and Proposition~\ref{prop:Properties_of_pregroups}.




Below we will make repeated use of Lemma~\ref{lem:alt_sigma}(i) and the defining properties of pregroups without mention.  
\begin{lemma}\label{lem:sigma_preserves_linear_negations}
For all $U\in \mathsf{Up}\left(P, \leqslant\right)$, we have $\sigma\left(-U\right) = -\sigma\left(U\right)$ and $\sigma\left({\sim}U\right) = {\sim}\sigma\left(U\right)$.  
\end{lemma}

\begin{proof}
Let $\left(x, y\right) \in \sigma\left(-U\right)$. Then there is some $w \in -U$ such that $xw\leqslant y$. By definition of $-U$, $w=v^{\ell}$ with $v \notin U$. 
Hence $v^{\ell} \leqslant x^rxv^{\ell} \leqslant x^ry$.
Now suppose $\left(x, y\right) \notin -\sigma\left(U\right)$. 
Then $\left(x, y\right) \notin \alpha\, ;\sigma\left(U\right)^{c\smile} = \left(\alpha\, ; \sigma\left(U\right)^\smile\right)^c$, and so $\left(x, y\right) \in \alpha\, ; \sigma\left(U\right)^\smile$. From this we obtain $\left(y, \alpha\left(x\right)\right) \in \sigma\left(U\right)$, i.e., $\left(y, x^{rr}\right) \in \sigma\left(U\right)$. 
Hence there exists $u \in U$ such that $yu \leqslant x^{rr}$, and so $x^ryu \leqslant x^rx^{rr} \leqslant 1$. From $v^{\ell} \leqslant x^{r}y$ we get $v^{\ell}u \leqslant x^ryu$, so $v^{\ell}u\leqslant 1 $, which gives $u=1\cdot u \leqslant vv^{\ell}u \leqslant v$, a contradiction since $U$ an up-set and $v \notin U$. Thus $(x,y) \in -\sigma(U)$.

Conversely, let $\left(x, y\right) \in -\sigma\left(U\right)$. Then $\left(x, y\right) \in \alpha\mathbin{;} \sigma\left(U\right)^{c\smile}$, and therefore $\left(x^{rr}, y\right) \in \sigma\left(U\right)^{c\smile}$, which means $\left(y, x^{rr}\right) \notin \sigma\left(U\right)$. This implies that for all $u \in U$, $yu \nleqslant x^{rr}$. To show  $\left(x, y\right) \in \sigma\left(-U\right)$, we must find $w\in -U$ such that $xw \leqslant y$. Recall $w \in -U$ implies  $w=v^{\ell}$ where $v \notin U$. We have $yy^{r}x^{rr}\leqslant x^{rr}$ so $y^rx^{rr} \notin U$. Hence $(y^{r}x^{rr})^{\ell}=x^{rr\ell}y^{r\ell}=x^ry \in - U$. Now $xx^ry\leqslant 1\cdot y =y$, so $(x,y)\in \sigma(-U)$. 

Now let $(x, y) \in \sigma({\sim}U)$. Then there is some $w \in {\sim}U$ such that $xw \leqslant y$. 
By definition of ${\sim}U$, $w=v^{\sim}$ with $v \notin U$. We get $y^{\ell}xv^{r} \leqslant y^{\ell}y \leqslant 1$. Then suppose $(x, y) \notin {\sim}\sigma(U) = \sigma(U)^{c\smile}\mathbin{;}\alpha = (\sigma(U)^\smile\mathbin{;} \alpha)^c$. It thus follows that $(x, y) \in \sigma(U)^\smile\mathbin{;}\alpha$, which means $(\alpha^{-1}(y), x) \in \sigma(U)$, i.e., $(y^{\ell\ell}, x) \in \sigma(U)$. 
So, there exists $u \in U$ such that $y^{\ell\ell}u\leqslant x$. Now 
$u \leqslant y^{\ell}y^{\ell \ell}u \leqslant y^{\ell}x$. Combining with the earlier inequality we get $uv^{r}\leqslant 1$ and thus $u\leqslant uv^{r}v\leqslant v$. Since $U$ an up-set, this contradicts $v \notin U$. Hence  $(x,y)\in {\sim}\sigma(U)$.

Conversely, let $\left(x, y\right) \in {\sim}\sigma\left(U\right) =  \sigma\left(U\right)^{c\smile} \mathbin{;} \alpha$. Then $\left(x, y^{\ell\ell}\right) \in \sigma\left(U\right)^{c\smile}$, and so $\left(y^{\ell\ell}, x\right) \notin \sigma\left(U\right)$. 
Hence for all $u \in U$, $y^{\ell \ell}u \nleqslant x$.
To show  $(x, y) \in \sigma({\sim}U)$, we have to find some 
$w \in {\sim}U$ such that $xw \leqslant y$. Any such $w$ is equal to $v^r $ where $v \notin U$. Now $y^{\ell\ell}y^{\ell}\leqslant 1$ so $y^{\ell\ell}y^{\ell}x \leqslant x$. Hence $y^{\ell}x \notin U$, so 
$(y^{\ell}x)^r=x^{r}y^{\ell r}=x^ry \in {\sim}U$ and can be used as the required $w$ since $xx^{r}\leqslant 1$. 
\qed
\end{proof}


The above lemma combines with Theorem~\ref{thm:sigma-pres-meets}(iv) to show that $\mathcal{D}(\mathbf{P})$ 
is representable in the sense of Definition~\ref{Definition:RDqRA}. 

\begin{theorem}\label{thm:D(P)_representable}
Let $\mathbf P = \langle P, \leqslant, \cdot, 1, ^{\ell}, ^r\rangle$ be a pregroup. Then $\mathcal{D}(\mathbf{P})$ is representable. 
Moreoever, every pregroup DInFL-algebra is representable.  
\end{theorem}

Let $\mathbf{P}$ be an ipo-monoid. We remark that $\mathbf{P}$ being a pregroup is not a necessary condition for $\mathcal{D}(\mathbf{P})$ to be representable. Consider the two-element ipo-monoid with $0<1$. Then $\mathcal{D}(\mathbf{P})$ is isomorphic to the three-element Sugihara chain $\mathbf{S}_3$, which is representable~\cite[Example 5.1]{RDqRA25}. 

The proposition below shows that preservation of $-$ and ${\sim}$ by $\sigma$ happens exactly when $\mathbf{P}$ is a pregroup. 
\begin{proposition}\label{prop:sig_pres_linneg_implies_pregroup}
Let $\mathbf P = \langle P, \leqslant, \cdot, 1, ^-, ^\sim\rangle$ be an ipo-monoid. If the function $\sigma: \mathsf{Up}\left(P, \leqslant\right) \to \mathsf{Up}\left(P^2, \preccurlyeq\right)$ preserves $\sim$ and $-$, then $\mathbf{P}$ is a pregroup. 
\end{proposition}

\begin{proof}
Suppose $\mathbf P$ is not a pregroup.  
For each possible failure of a pregroup axiom, we show that $\sigma$ fails to preserve either $-$ or ${\sim}$. 

\underline{Case 1:} there is some $x \in P$ such that $x^-x \not\leqslant 1$.
By Lemma~\ref{lem:alt_sigma}(ii) we get that $(1,1) \notin \sigma\left({\uparrow}x^-x\right)$. 
It follows 
that $(1,1) \in \sigma\left({\uparrow}x^-x\right)^{c\smile}$. We also have $(1,1) = (1, 1^{\sim\sim}) \in \alpha$, so $(1, 1) \in \alpha\mathbin{;}\sigma\left({\uparrow}x^-x\right)^{c\smile} = -\sigma\left({\uparrow}x^-x\right)$. 

Now suppose 
$(1,1) \in \sigma\left(-{\uparrow}x^-x\right)$. Hence, by Lemma~\ref{lem:alt_sigma}(i), there exists some 
$v \in -{\uparrow}x^-x$ such that 
$v = 1\cdot v \leqslant 1$. 
This implies that $v 0 = v1^{\sim} \leqslant 0$. Now $x^-x \leqslant 0$, so $vx^-x \leqslant v0$, and therefore $vx^-x \leqslant 0$. It follows that $v^{\sim -}x^-x \leqslant 0$, and consequently $x^-x \leqslant v^{\sim}$. Thus, $v^\sim \in {\uparrow}x^-x$, and hence $v= v^{\sim -} \notin -{\uparrow}x^-x$, which is a contradiction. 
Hence $(1,1) \notin \sigma\left(-{\uparrow}x^-x\right)$ and 
$-\sigma\left({\uparrow}x^-x\right) \neq \sigma\left(-{\uparrow}x^-x\right)$. 


\underline{Case 2:} there exists $x \in P$ with $1 \nleqslant xx^-$. Let $U=P{\setminus}{\downarrow}xx^-$. By definition, $(xx^-)^{\sim} \in {\sim}U$. From $xx^-\leqslant (xx^-)^{\sim -}$ and item (ii) in Proposition~\ref{prop:alternative_def_ipo}, we get $x^-(xx^-)^{\sim}\leqslant x^{\sim}$, and so by Lemma~\ref{lem:alt_sigma}(i) we have $(x^-,x^{\sim})\in\sigma({\sim}U)$. 

Since $1 \in U$ we get $(x^-,x^-)\in \sigma(U)$, and therefore $(x^-,x^-)\notin \sigma(U)^{c\smile}$. If $(x^-,x^{\sim})\in
\sigma(U)^{c\smallsmile}\mathbin{;}\alpha$ then since $\alpha(x^-)=x^{\sim}$, we would get $(x^-,x^-) \in \sigma(U)^{c\smallsmile}$, a contradiction. Hence 
$(x^-,x^{\sim}) \notin 
\sigma(U)^{c\smallsmile}\mathbin{;}\alpha$ and so $\sigma({\sim} U)\neq {\sim}\sigma(U)$. 

\underline{Case 3:} there is some $x \in P$ such that  $xx^{\sim} \nleqslant 1$. Using  Lemma~\ref{lem:alt_sigma}(ii) we can get $(1,1) \in 
\sigma({\uparrow}xx^{\sim})^{c\smallsmile}\mathbin{;}\alpha=-\sigma({\uparrow}xx^{\sim})$. 

Suppose $(1,1)\in \sigma(-{\uparrow}xx^{\sim})$. By Lemma~\ref{lem:alt_sigma}(i) there exists $v \in -{\uparrow}xx^{\sim}$ such that $v=1\cdot v \leqslant 1$. By definition, $v=u^{-}$ where $xx^{\sim} \nleqslant u$. By Lemma~\ref{lem:properties_ipo-mon}(ii)  we get $u^- \nleqslant (xx^{\sim})^-$ and by Lemma~\ref{lem:properties_ipo-mon}(iv) $u^-x\nleqslant x$. This contradicts the fact that $u^-\leqslant 1$, hence $(1,1) \notin \sigma(-{\uparrow}xx^{\sim})$. 

\underline{Case 4:} there is some $x \in P$ with $1 \nleqslant x^{\sim}x$. Let $U=P{\setminus}{\downarrow}x^{\sim}x$. The argument is similar to Case 2. Since $1 \in U$, we will get $(x,x^{\sim\sim})\notin {\sim}\sigma(U)$. Applying (ii) from Proposition~\ref{prop:alternative_def_ipo} to $x^{\sim}x\leqslant (x^{\sim}x)^{\sim -}$ gets us 
$x (x^{\sim}x)^{\sim} \leqslant x^{\sim\sim}$ and since $(x^{\sim}x)^{\sim}\in {\sim}U$ we have $(x,x^{\sim\sim})\in \sigma({\sim}U)$.
\qed
\end{proof}


\section{Representations of DqRAs via ortho pregroups}\label{sec:DqRA-ortho}

In order to extend our construction from DInFL-algebras to distributive quasi relation algebras, we need the ipo-monoid to be equipped with an additional unary operation satisfying certain conditions. 


\begin{definition}\label{def:ortho ipomonoid}
An \emph{ortho ipo-monoid} is a structure $\mathbf{P}=\langle P,\leqslant, \cdot, 1, ^-, ^{\sim}, ^\neg\rangle$ such that $\langle P,\leqslant, \cdot, 1, ^-, ^{\sim}\rangle$ is an ipo-monoid and the following 
hold for all $x, y, z\in P$:  
\begin{multicols}{2}
\begin{enumerate}[\normalfont (i)]
\item $x^{\neg \neg} = x$, and 
\item $xy\leqslant z^-$ iff $y^{\sim\neg}x^{\sim\neg}\leqslant z^\neg$.
\end{enumerate} 
\end{multicols}
\end{definition}
We remark that (ii) above is equivalent (by (ii) of Proposition~\ref{prop:alternative_def_ipo})
to $x y \leqslant z^{\sim}$ iff $x^{\sim\neg}z^{\sim\neg}\leqslant y^\neg$.
We will use the following properties of an ortho ipo-monoid without referring to them. 

\begin{proposition}\label{prop:properties_ortho_ipomonoids}
Let $\mathbf{P}=\langle P,\leqslant, \cdot, 1, ^-, ^{\sim}, ^\neg\rangle$ be an ortho ipo-monoid. Then the following hold for all $x, y \in P$:
\begin{enumerate}[\normalfont (i)]
\item $1^\neg = 1^{-} = 1^{\sim}$, 
\item $x \leqslant y$ iff $y^{\neg} \leqslant x^\neg$, and 
\item $x^{\sim \neg} = x^{\neg -}$. 
\end{enumerate}
\end{proposition}

\begin{proof}
For item (i), it suffices to prove that $1^\neg = 1^{\sim}$ since $1^- = 1^{\sim}$ holds in any ipo-monoid.  We have $1^{-\sim \neg}\leqslant 1^{\neg}$, and so $1^{-\sim\neg} \cdot 1^{\neg -\sim \neg} = 1^{-\sim\neg} \cdot 1 \leqslant 1^{\neg}$. 
Applying Definition~\ref{def:ortho ipomonoid}(ii) to this gives $1^{\neg -} \cdot 1^- \leqslant 1^-$. 
Hence, by  Proposition~\ref{prop:alternative_def_ipo}(ii), $1^- = 1^-\cdot 1 \leqslant 1^{\neg -\sim} = 1^{\neg}$.

We also have $1^{\neg -}\cdot 1 \leqslant 1^{\neg -}$, and so, by Definition~\ref{def:ortho ipomonoid}(ii), $1^{\sim \neg}\cdot 1^{\neg -\sim \neg} \leqslant 1^{\neg \neg}$, i.e., $1^{\sim \neg} = 1^{\sim\neg}\cdot 1 \leqslant 1$. Hence, $1^{\sim\neg} \cdot 1^{\sim\neg} \leqslant 1^{\sim \neg}\cdot 1$, and combining this with $1^{\sim\neg}\cdot 1 \leqslant 1$ gives $1^{\sim\neg} \cdot 1^{\sim\neg} \leqslant 1 = 1^{\neg\neg}$. Another application of  Definition~\ref{def:ortho ipomonoid}(ii) yields $1 \cdot 1 \leqslant 1^{\neg -}$. Therefore $1 \leqslant 1^{\neg -}$, and so $1^\neg = 1^{\neg - \sim} \leqslant 1^{\sim} = 1^{-}$. We can thus conclude that $1^{\neg}= 1^- = 1^{\sim}$.   

For (ii), assume $x \leqslant y$. Then $y^- \leqslant x^-$ by Lemma~\ref{lem:properties_ipo-mon}(ii), and so $1\cdot y^- \leqslant x^-$. Hence, applying Definition~\ref{def:ortho ipomonoid}(ii) to this yields $y^{-\sim\neg}\cdot 1^{\sim\neg} \leqslant x^{\neg}$ or, equivalently, $y^{\neg}\cdot 1^{\sim\neg}\leqslant x^{\neg}$. Since $1^{\neg} = 1^{\sim}$ by (i) and $1^{\neg\neg} =1$, we get $1^{\sim \neg} =1$, and thus $y^{\neg} \leqslant x^{\neg}$. 

For item (iii), we first show that $x^{\sim \neg} \leqslant x^{\neg -}$. We have $x^{-\sim}\cdot 1\leqslant x^{-\sim}$, and so $x^-x^{-\sim}\leqslant 1^-$ by Proposition~\ref{prop:alternative_def_ipo}(ii). Applying (ii) of Definition~\ref{def:ortho ipomonoid} gives $x^{-\sim\sim\neg}x^{-\sim \neg} \leqslant 1^\neg$. This is equivalent to $x^{\sim\neg}x^{\neg} \leqslant 1^{\neg-\sim}$. Applying Proposition~\ref{prop:alternative_def_ipo}(ii) again gives $1^{\neg -}\cdot x^{\sim\neg} \leqslant x^{\neg -}$, and therefore, by (ii) of Definition~\ref{def:ortho ipomonoid}, $x^{\sim\neg\sim\neg}\cdot 1^{\neg -\sim\neg} \leqslant x^{\neg\neg}$, i.e., $x^{\sim\neg \sim \neg}\leqslant x$. Consequently, by (ii), $x^{\neg} \leqslant x^{\sim \neg\sim\neg\neg} = x^{\sim\neg \sim}$, and thus $x^{\sim\neg} = x^{\sim\neg \sim -} \leqslant x^{\neg -}$. 

We also have $x\cdot 1\leqslant x^{\neg\neg}$, which implies $x^{\neg-\sim \neg} \cdot 1^{\neg-\sim\neg} \leqslant x^{\neg\neg}$. Hence, by Definition~\ref{def:ortho ipomonoid}(ii), $1^{\neg -}\cdot x^{\neg -} \leqslant x^{\neg -}$, and so, $x^{\neg -}x^{\neg} \leqslant 1^{\neg - \sim} = 1^{\neg}$. This is equivalent to $x^{\neg -\neg -\sim\neg} x^{\neg\neg -\sim \neg} \leqslant 1^{\neg}$, so another application of Definition~\ref{def:ortho ipomonoid}(ii) yields $x^{\neg\neg -} x^{\neg -\neg -} \leqslant 1^-$. Therefore, $x^{\neg -\neg -} = x^{\neg -\neg -} \cdot 1\leqslant x^{\neg\neg -\sim} = x$, and consequently $x^{\sim} \leqslant x^{\neg-\neg -\sim} = x^{\neg -\neg}$. Thus, by (ii), $x^{\neg -} = x^{\neg-\neg\neg} \leqslant x^{\sim \neg}$. 
\qed
\end{proof}

If $\langle P, \leqslant, \cdot, 1, ^\ell, ^r\rangle$ is a pregroup and $^\neg: P \to P$ satisfies items (i) and (ii) of Definition~\ref{def:ortho ipomonoid}, then we will refer to $\mathbf P = \langle P, \leqslant, \cdot,1, ^{\ell}, ^r, ^{\neg}\rangle$ as an \emph{ortho pregroup}. 

\begin{example}
If $\mathbf{G}=\langle G,\cdot,e, ^{-1}\rangle$ is a group and $\mathrm{id}_G:G\to G$ the identity function, then $\langle G, =, \cdot ,e,^{-1},^{-1},\mathrm{id}_G\rangle$ is an ortho pregroup. 
If $\mathbf{G}$ is an Abelian group, then
$\langle G, =, \cdot ,e,^{-1},^{-1},^{-1}\rangle$ is another example of an ortho pregroup. 
\end{example}

We now construct a DqRA from an ortho ipo-monoid 
$\mathbf{P}=\langle P, \leqslant, \cdot, 1, ^-, ^{\sim}, ^\neg\rangle$.  For each $U \in \mathsf{Up}(P,\leqslant)$, define $\neg U=\{ x^\neg \mid x \notin U\}$. 

\begin{lemma}\label{lem:neg_upset}
Let $\mathbf{P} = \langle P, \leqslant, \cdot, 1, ^-, ^{\sim}, ^\neg\rangle$ be an ortho ipo-monoid. If  $U \in \mathsf{Up}\left(P, \leqslant\right)$, then ${\neg}U \in \mathsf{Up}\left(P, \leqslant\right)$.
\end{lemma}

\begin{proof}
Let $U \in \mathsf{Up}\left(P, \leqslant\right)$, and let $x\in {\neg}U$ and $y \in P$. Assume $x \leqslant y$. Since $x \in {\neg}U$, it follows that $x^{\neg\neg} \in {\neg}U$, and hence $x^{\neg} \notin U$. Since $x\leqslant y$, we have $y^\neg \leqslant x^\neg$ by item (ii) of Definition~\ref{def:ortho ipomonoid}. Hence, since $U$ is an upset of $\langle P, \leqslant\rangle$, it follows that $y^\neg \notin U$. This shows that $y = y^{\neg\neg} \in {\neg}U$.  
\qed
\end{proof}

Since $\langle \mathsf{Up}(P,\leqslant), \cap, \cup, \bullet, {\uparrow}1, -, {\sim}\rangle$ is a DInFL-algebra by Theorem~\ref{thm:DInFL_from_ipo-monoid}, it
will follow that $\langle \mathsf{Up}(P,\leqslant), \cap, \cup, \bullet, {\uparrow}1, -, {\sim}, \neg\rangle$ is a DqRA if we can show that for all $U \in \mathsf{Up}(P,\leqslant)$, we have $\neg\neg U = U$, and that (\textsf{Dm}) and (\textsf{Dp}) hold. 

\begin{theorem}\label{thm:DqRA_from_ortho-ipo-monoid}
Let $\mathbf{P}=\langle P, \leqslant, \cdot, 1, ^-, ^{\sim}, ^\neg\rangle$ be an ortho ipo-monoid.  Then the algebra
$\mathcal{Q}(\mathbf{P}) = \langle \mathsf{Up}(P,\leqslant), \cap, \cup, \bullet, {\uparrow}1, -, {\sim}, \neg \rangle$ is a DqRA.
\end{theorem}

\begin{proof}

First, $x \in \neg\neg U$ iff $x^\neg \notin \neg U$ iff $x^{\neg\neg} \in U$ iff $x \in U$. 

Next, we show that (\textsf{Dm}) holds. Let $U, V \in \mathsf{Up}(P,\leqslant)$. Then $x \in \neg(U \cap V)$ iff $x^\neg \notin U \cap V$ iff $x^\neg \notin U$ or $x^\neg \notin V$ iff $x \in \neg U$ or $x\in \neg V$ iff $x \in \neg U \cup \neg V$. 

Finally, we show that (\textsf{Dp}) holds. 
To see that $\neg(U\bullet V) \subseteq {\sim}(-\neg V \bullet -\neg U)$, suppose $x \notin {\sim}(-\neg V \bullet -\neg U)$. Then we have $x^{-\sim} \notin {\sim}(-\neg V \bullet -\neg U)$, and so $x^- \in (-\neg V \bullet -\neg U)$. Hence, there exist $v \in -\neg V$ and $u \in -\neg U$ such that $vu\leqslant x^-$. Consequently, $u^{\sim\neg}v^{\sim\neg}\leqslant x^\neg$ by item (iii) of Definition~\ref{def:ortho ipomonoid}. From $v \in -\neg V$ and $u \in -\neg U$, we get $v^{\sim\neg} \in V$ and $u^{\sim\neg} \in U$. It thus follows that $x^\neg \in U\bullet V$, which means $x = x^{\neg\neg} \notin \neg(U\bullet V)$.

Now suppose $x \notin \neg(U \bullet V)$. Then $x^\neg  \in U\bullet V$, and so, there exist $u \in U$ and $v \in V$ such that  $uv \leqslant x^\neg$. Hence, by  Definition~\ref{def:ortho ipomonoid}(iii), $v^{\sim\neg}u^{\sim\neg} \leqslant x^\neg$.
\qed
\end{proof}

We say that a DqRA $\mathbf{A}$ is an \emph{ortho pregroup DqRA} if $\mathbf{A}$ is isomorphic to a subalgebra of $\mathcal{Q}(\mathbf{P})$ for some ortho pregroup $\mathbf{P}$.  

To build a DqRA of binary relations from an ortho ipo-monoid, define a map $\beta: P \to P$ by setting, for all $x \in P$, $\beta(x) = x^\neg$. Then, by Proposition~\ref{prop:properties_ortho_ipomonoids}, $\beta$ and $\alpha(x)=x^{\sim\sim}$ satisfy the conditions of Theorem~\ref{Theorem:Dq(E)}. 
\begin{theorem}\label{thm:binrel-DqRA-from-ortho-pregroup}
Let $\mathbf P = \langle P, \leqslant, \cdot, 1, ^-, ^{\sim}, ^\neg\rangle$ be an ortho ipo-monoid. Then the algebra $\left\langle \mathsf{Up}\left(P^2, \preccurlyeq\right), \cap, \cup, \mathbin{;}, \leqslant, -, \sim, \neg \right\rangle$ is a DqRA. 

\end{theorem}

Recall the map $\sigma: \mathsf{Up}(P,\leqslant) \to \mathsf{Up}(P^2,\preccurlyeq)
$ from Section~\ref{sec:embedDRL}. 
Below we show that $\sigma$ preserves the $\neg$ operation of $\mathcal{Q}(\mathbf{P})$ when $\mathbf{P}$ is an ortho pregroup.

\begin{lemma}
Let $\mathbf{P}=\langle P,\leqslant,\cdot, 1, ^{\ell},^{r},^{\neg}\rangle$ be an ortho pregroup. For  $U\in \mathsf{Up}\left(P, \leqslant\right)$, we have $\sigma\left(\neg U\right) = \neg \sigma\left(U\right)$.  
\end{lemma}

\begin{proof}
Let $\left(x, y\right) \in \sigma\left(\neg U\right)$. Then there is some $u \in \neg U$ such that $xu \leqslant y$. 
Now suppose $\left(x, y\right) \notin \neg \sigma\left(U\right)$. Then we have $\left(x, y\right) \notin \alpha\, \mathbin{;}\beta\mathbin{;} \sigma\left(U\right)^{c}\mathbin{;} \beta = \left(\alpha\mathbin{;} \beta\mathbin{;} \sigma\left(U\right)\mathbin{;}\beta\right)^c$, and so $\left(x, y\right) \in \alpha\mathbin{;}\beta\mathbin{;} \sigma\left(U\right)\mathbin{;}\beta$. Hence, $\left(x^{rr\neg}, y^\neg\right) \in \sigma\left(U\right)$, which means there is some $v \in U$ such that $x^{rr}v \leqslant y^{\neg}$. This is equivalent to $x^{rr\neg} v^{\neg \ell r \neg} \leqslant y^{\neg}$. Applying Definition~\ref{def:ortho ipomonoid}(ii) to this gives $v^{\neg\ell} x^r \leqslant y^{\ell}$. Applying Proposition~\ref{prop:Properties_of_pregroups} to $xu\leqslant y$ gives $y^{\ell} \leqslant (xu)^{\ell} = u^{\ell}x^{\ell}$. We thus obtain $v^{\neg \ell}x^r \leqslant u^{\ell}x^{\ell}$, and thus $v^{\neg \ell}x^rx \leqslant u^{\ell}x^{\ell}x$. Now $1 \leqslant x^r x$ and $x^{\ell}x \leqslant 1$, so 
$v^{\neg \ell} \leqslant u^{\ell}$ or, equivalently, $v
\leqslant u^{\neg}$. 
Therefore, since $v \in U$ and $U$ is an upset of $\mathsf{Up}\left(P, \leqslant\right)$, it must be the case that $u^\neg \in U$. This gives $u=u^{\neg\neg} \notin \neg U$, which is a contradiction.


For the other containment, let $(x, y)\in \neg\sigma\left(U\right) = \alpha\mathbin{;} \beta\mathbin{;} \sigma\left(U\right)^c\mathbin{;} \beta$. Then $\left(x^{rr\neg}, y^\neg\right) \in \sigma\left(U\right)^c$, and so $\left(x^{rr\neg}, y^\neg\right) \notin \sigma\left(U\right)$. Consequently, for all $u \in P$, if $u \in U$, then $x^{rr\neg} \not\leqslant y^{\neg}$.
We have to find $v \in \neg U$ such that $xv \leqslant y$. 
Now $(x^ry)^{\ell}x^ry \leqslant 1$, so $(x^ry)^{\ell}x^ryy^{\ell} \leqslant y^{\ell}$. 
We also have $1\leqslant yy^{\ell}$, which means $(x^ry)^{\ell}x^r\leqslant (x^ry)^{\ell}x^ryy^{\ell}$. It follows that $(x^ry)^{\ell}x^r \leqslant y^{\ell}$. Applying (ii) of Definition~\ref{def:ortho ipomonoid} to this gives $x^{rr\neg}(x^ry)^{\ell r\neg}\leqslant y^{\neg}$ or, equivalently, $x^{rr\neg}(x^ry)^{\neg}\leqslant y^{\neg}$. 
This implies $\left(x^r y\right)^\neg \notin U$, and therefore $x^ry \in \neg U$. If we can thus show that $x x^r y \leqslant y$, we would be done. Indeed, $x  x^r \leqslant 1$, which means $x  x^r y \leqslant y$. 
%
\qed
\end{proof}

Combining the above lemma with the results of Section~\ref{sec:Rep_DInFL}, we obtain the following representability result for ortho pregroup DqRAs. 

\begin{theorem}\label{thm: pre-group_representable_DqRAs}
Let $\mathbf{P} = \langle P, \leqslant, \cdot, 1, ^\ell, ^r, ^\neg\rangle$ be an ortho pregroup. Then $\mathcal{Q}(\mathbf{P})$ is representable. 
Moreover, every ortho pregroup DqRA is representable. 
\end{theorem}
\section{Examples from products of groups}\label{sec:apps}


In this section we show that the class of algebras in Example~\ref{ex:represenatble_via_Z7} is ortho pregroup representable and therefore representable by Theorem~\ref{thm: pre-group_representable_DqRAs}. The case for $\mathbf{A}_3$ can also be found in~\cite{JS23}.

\begin{theorem}\label{thm:A_n_pregroup_representable}
Let $n \geqslant 3$. Then the algebra $\mathbf{A}_n= \langle A_n, \wedge, \vee, \cdot, 1, -, \sim, \neg\rangle$ is an ortho pregroup DqRA. 
\end{theorem}

\begin{proof}
First consider the algebra $\mathbf{A}_3 = D_{4,8}^6$ in Figure~\ref{fig:An_for_n=3,4} and the ortho pregroup $\mathbf{Z}_7=\langle \mathbb{Z}_7, =, +,0,  -, -, -\rangle$. 
The algebra $\mathbf{A}_3$ is isomorphic to a subalgebra of the upset algebra $\mathcal{Q}(\mathbf{Z}_7)$. Let $V_{-1} = \varnothing$, $V_0 = \{1, 2, 4\}$, $V_1 = \mathbb{Z}_7{\setminus} \{0\}$, $U_{-1} = \{0\}$, $U_0 = V_0\cup \{0\} = \{0,1, 2, 4\}$ and $U_1 = \mathbb{Z}_7$. Define a map $\psi: A_n \to \textsf{Up}(\mathbb{Z}_7, =)$ by setting, for all $a_i, b_j \in A_n$,  $\psi(a_i) = U_i$ and $\psi(b_j) = V_j$. Then $\psi$ is an embedding from $\mathbf{A}_3$ into $\mathcal{Q}(\mathbf{Z}_7)$. 

Now let $n \geqslant 4$ and consider the ortho pregroup $$\mathbf{Z}_{n-2}\times \mathbf{Z}_7 =\langle\mathbb{Z}_{n-2}\times \mathbb{Z}_7, =, +,(0,0), -,-,-\rangle.$$ Assume first $n = 2k$ where $k \geqslant 2$. Set $$T = \{(0,1), (0,2), (0,4)\} \cup \{(m, \ell) \mid 1 \leqslant m \leqslant n-3, \ell\in\{3, 5, 6\}\}.$$  For $j \in \{0, 1\ldots, k-1\}$, define 
$$
V_{-k+j} = 
\begin{cases} 
\varnothing & \text{if }\, j = 0 \\ 
T  & \text{if }\, j = 1\\ 
V_{-k+1} \cup \{(m, 0)\mid 1\leqslant m \leqslant j-1\} & \text{if }\, 2 \leqslant j \leqslant k-1
\end{cases}
$$
and, for $j \in \{0, 1, \ldots, k-1\}$, define 
$$
V_{k-j} = 
\begin{cases}
V_{-1} \cup \{(m, 0)\mid k-1\leqslant m \leqslant 2k-2-j\} & \text{if } \, 1\leqslant j \leqslant k-1\\
\mathbb{Z}_{n-2}\times\mathbb{Z}_7 {\setminus} \{(0,0)\} & \text{if } j= 0.
\end{cases}
$$
Finally, for $i \in \{-k, \ldots, -1, 1, \ldots, k\}$, define 
$
U_{i} = V_{i} \cup \{(0, 0)\}.
$
%
%

Now assume $n = 2k+ 1$ where $k \geqslant 2$. For $j \in \{0, 1\ldots, k\}$, define 
$$
V_{-k+j} = 
\begin{cases} 
\varnothing & \text{if }\, j = 0 \\ 
T  & \text{if }\, j = 1\\ 
V_{-k+1} \cup \{(m, 0)\mid 1\leqslant m \leqslant j-1\} & \text{if }\, 2 \leqslant j \leqslant k
\end{cases}
$$
and, for $j \in \{0, 1, \ldots, k-1\}$, define 
$$
V_{k-j} = 
\begin{cases}
V_{0} \cup \{(m, 0)\mid k\leqslant m \leqslant 2k-1-j\} & \text{if } \, 1\leqslant j \leqslant k-1\\
\mathbb{Z}_{n-2}\times\mathbb{Z}_7 {\setminus} \{(0,0)\} & \text{if } j= 0.
\end{cases}
$$
For $i \in \{-k, \ldots, -1, 0, 1, \ldots, k\}$, define 
$
U_{i} = V_{i} \cup \{(0, 0)\}.
$
In both cases it can be shown that the map $\psi: A_n \to \textsf{Up}(\mathbb{Z}_{n-2} \times\mathbb{Z}_7, =)$ defined by   $\psi(a_i) = U_i$ and $\psi(b_j) = V_j$ for all $a_i, b_j \in A_n$ is an embedding from $\mathbf{A}_n$ into $\mathcal{Q}(\mathbf{Z}_{n-2}\times \mathbf{Z}_7)$. 
\qed
\end{proof}

\section{Future work}\label{sec:future}

The results in this paper provide an opportunity to further explore representability of DInFL-algebras and DqRAs along two interesting paths.


The first of these paths  is to consider a partially-defined operation on the underlying pregroup. 
J\'{o}nsson and Tarski~\cite[Section 5]{JT52} observe that the complex algebras of Brandt groupoids are relation algebras and that there is a one-to-one correspondence between representable relation algebras and subalgebras of the complex algebras of Brandt groupoids~\cite[Theorem 5.8]{JT52}. 
Jipsen~\cite[Section 4]{Jip17} described  \emph{partially ordered groupoids} and showed that the collection of downsets can be equipped with additional algebraic structure. We conjecture that there is an appropriate  definition of \emph{pregroupoids}, where the binary operation would be partially defined, and that these structures would provide a fruitful setting for investigating representability of DInFL-algebras. 

Currently, it is not known whether there are non-representable DInFL-algebras or DqRAs. The smallest known algebra without a representation is the three-element MV-chain \L$_3$. The second path we wish to explore is to modify the construction of Givant and Andr\'{e}ka~\cite{GA02,Giv18} which has been used to construct non-representable relation algebras via so-called \emph{coset relation algebras}~\cite{AG18,ANG20}. We hope to be able to use our representations via group-like structures to obtain non-representability results.


\end{document}